\documentclass[10.5pt]{article}

\usepackage{epsfig,amsmath,amssymb,bm}
\usepackage[legacycolonsymbols]{mathtools}
\usepackage{typearea}\typearea{13}
\usepackage{url}

\usepackage{graphicx}%
\usepackage{multirow}%
\usepackage{amsmath,amssymb,amsfonts}%
\usepackage{amsthm}%
\usepackage{mathrsfs}%
\usepackage[title]{appendix}%
\usepackage{xcolor}%
\usepackage{textcomp}%
\usepackage{manyfoot}%
\usepackage{booktabs}%
\usepackage{algorithm}%
\usepackage{algorithmicx}%
\usepackage{algpseudocode}%
\usepackage{listings}%
\newtheorem{prop}{Prop}[section]

\newtheorem{lemma}[prop]{Lemma}

\newtheorem{theorem}[prop]{Theorem}

\makeatletter
\@addtoreset{equation}{section}
\makeatother

\makeatletter
\long\def\@makecaption#1#2{{\small
\advance\leftskip1cm
\advance\rightskip1cm
\vskip\abovecaptionskip
\sbox\@tempboxa{#1: #2}%
\ifdim \wd\@tempboxa >\hsize
 #1: #2\par
\else
\global \@minipagefalse
\hb@xt@\hsize{\hfil\box\@tempboxa\hfil}%
\fi
\vskip\belowcaptionskip}}
\makeatother
\def\eq#1\en{\begin{equation}#1\end{equation}}  
\def\eqa#1\ena{\begin{align}#1\end{align}}
\def\eqg#1\eng{\begin{gather}#1\end{gather}}
\newcommand{\lb}[1]{\label{e:#1}}
\newcommand{\rlb}[1]{\eqref{e:#1}} 
\newcommand{\nl}{\notag\\}

\newcommand{\up}{\uparrow}

\newcommand{\bs}{\backslash}
\newcommand{\sumtwo}[2]%
{\mathop{\sum_{#1}}_{#2}}
\newcommand{\sumthree}[3]%
{\mathop{\mathop{\sum_{#1}}_{#2}}_{#3}}
\newcommand{\sumfour}[4]%
{\mathop{\mathop{\mathop{\sum_{#1}}_{#2}}_{#3}}_{#4}} 
\newcommand{\snorm}[1]{\Vert#1\Vert}

\newcommand{\sbkt}[1]{\langle#1\rangle}
\newcommand{\bbkt}[1]{\bigl\langle#1\bigr\rangle}

\newcommand{\bbE}{\mathbb{E}}

\newcommand{\bbR}{\mathbb{R}}
\newcommand{\bbZ}{\mathbb{Z}}

\newcommand{\bsb}{\boldsymbol{b}}
\newcommand{\bstau}{\boldsymbol{\tau}}
\newcommand{\bss}{\boldsymbol{\sigma}}
\newcommand{\bssa}{\boldsymbol{\sigma}_\mathrm{a}}
\newcommand{\bssb}{\boldsymbol{\sigma}_\mathrm{b}}
\newcommand\calB{\mathcal{B}}
\newcommand\calC{\mathcal{C}}

\newcommand\dcalB{\partial\mathcal{B}}
\newcommand\dcalC{\partial\mathcal{C}}

\newcommand{\La}{\Lambda}
\newcommand{\LaL}{\Lambda_L}
\newcommand{\dLaL}{\partial\Lambda_L}
\newcommand{\Lal}{\Lambda_\ell}
\newcommand{\Laa}{\Lambda_\mathrm{a}}
\newcommand{\Lab}{\Lambda_\mathrm{b}}

\newcommand{\limsupL}{\limsup_{L\up\infty}}
\newcommand{\liminfL}{\liminf_{L\up\infty}}
\newcommand{\limL}{\lim_{L\up\infty}}

\newcommand{\qbr}{q_{\rm br}}
\newcommand{\qEA}{q_{\rm EA}}

\begin{document}

%%%%the following title part is used when making the submission version.
\title{ 
Self-averaging 
 of replica overlaps\\ in the random field 
Edwards-Anderson model}
\author{C. Itoi$^1$ and Y. Sakamoto $^2$
\\
\\
$^1$ Department of Physics, Graduate School of Science $\&$ Technology,
Nihon University, \\
Kanda-Surugadai, Chiyoda, Tokyo 101-8308, Japan\\
\\
$^2$ Laboratory of Physics, College of Science $\&$ Technology,
 Nihon University, \\
Narashinodai, Funabashi-city, Chiba, 274-8501,Japan }
\maketitle

\abstract{The self-averaging of the replica overlap is proven in the Edwards-Anderson (EA) model under random field almost everywhere in the coupling constant space in any dimension. The EA order parameter is represented in terms of the derivative of the free energy density with respect to the random field strength, regardless of boundary conditions. Tasaki's correlation inequality for finite-dimensional spin glass models shows that the expectation of the squared replica overlap is bounded by the squared EA order parameter. These simple evaluations enable us to prove that the variance of the replica overlap vanishes in the infinite-volume limit. The self-averaging of the replica bond overlap is proven also in the EA model with Gaussian exchange interaction without random field. Short-range spin glass models have been shown to behave differently from mean-field spin glass models with RSB phase.}\\

\noindent
keywords: spin glass, short-range interactions, concentration, replica symmetry breaking, critical dimension

%%%%%%%% the end of preprint version

\tableofcontents
\newpage
%%%%%%%%%%%%%%%%%%
\section{Introduction}
\label{s:intro}

%On the other hand,  
The Edwards-Anderson (EA) model \cite{EA} has been studied extensively 
as a well-known spin glass model with random short-range interactions. 
Nonetheless, its complexity still leaves several intriguing questions unsolved.
 It is an important question whether the EA model has replica symmetry breaking (RSB) phase 
 in some higher dimension or not. Although this  question in statistical physics has been  discussed for more than four decades after the discovery of the Parisi formula \cite{P0} for the SK  model, there has been no clear answer for the EA model.
In statistical physics, there is an empirical prediction stating that 
a finite-dimensional spin model above its lower critical dimension has a topologically equivalent 
  phase diagram to that in the corresponding mean-field spin model.
Further, a model above its upper critical dimension exhibits exactly the same critical behavior 
as the corresponding mean-field spin model.
Since RSB phase exists in mean-field spin glass models \cite{P0,T1,T0,T}, for example 
the Sherrington-Kirkpatrick (SK) model \cite{SK},
this empirical prediction implies that the EA model has an RSB phase in sufficiently high dimension.
%There are two conflicting ideas in theoretical physics `the RSB picture' by Parisi \cite{P0}
%and `the droplet picture' by  Fisher and Huse \cite{FH} to understand the low temperature behavior of %finite-dimensional spin glass models. 
%Although the droplet picture is primarily motivated by the behavior of the three-dimensional EA spin glass, %it is formulated in general finite dimensions. 

%There have been several criticisms for the existence of RSB phase in finite-dimensional spin models. 
%Newman and Stein have partially excluded the existence of RSB phase in the EA model \cite{NS,NewmanStein1997,NewmanStein2002}. 
%They have claimed that a finite-dimensional spin glass model should have a pure Gibbs state, then the RSB picture is unnatural.  

%There are several rigorous results for RSB in low temperature region of finite-dimensional spin models
%with random interactions.
%Nishimori and Sherrington have shown that the replica overlap is concentrated at its expectation value
%on the Nishimori line \cite{N} lying outside the spin glass phase in the EA model \cite{NS1}.  

Recently, Chatterjee has proven that the variance of the replica overlap vanishes in
the random field Ising model in any dimension and almost everywhere in the coupling constant space \cite{Chatterjee2015}. While this result is accepted as a natural consequence, analogous to the concentration of magnetization in the ferromagnetic Ising model under a uniform field, this rigorous result has invalidated several published claims of self-averaging violation of the replica overlap and 
existence of RSB phase in the random field Ising model \cite{BD,MY}.
This theorem  is proven using the Fortuin-Kasteleyn-Ginibre (FKG) inequality \cite{FKG} and
the Ghirlanda-Guerra identities \cite{AC,GG}. 
While the Ghirlanda-Guerra identities are well-known to hold  universally in  spin systems
with Gaussian random interactions,  
the FKG inequality is valid only in the random field Ising model with non-negative exchange interactions. 

The Griffiths theorem for ferromagnetic spin systems \cite{Gff} has been extended to that for other
spin models, for example for quantum anti-ferromagnets \cite{KT}.    
 Griffiths-type theorems for finite-dimensional spin glass models have been obtained as an extension of the original theorem in \cite{I1}. 
 %Several upper bounds on the variance of the replica overlap have been obtained in terms of order parameters including the EA order parameter. 
 The EA order parameter is defined by the  thermal averaged replica overlap maximized over boundary condition. The EA order parameter is represented in terms of the derivative of free energy density with respect to the RSB perturbation parameter. The obtained Griffiths-type theorem claims that the finite variance of the replica overlap in the replica symmetric Gibbs state leads to the existence of spontaneous replica symmetry breaking with a finite EA order parameter. The theorem for finite-dimensional spin glass models has been proven by that representation and Tasaki's inequality \cite{HT}.
%, which enables us to prove that the expectation of the
%squared replica overlap is bounded by the EA order parameter from above in the EA model in any dimension.

  In the present paper, we study the Edwards-Anderson model under independent and identically distributed (i.i.d.) Gaussian random field. It is proven that the variance of the replica overlap converges to zero in the infinite-volume limit for almost all coupling constants in any dimension. Our result can be regarded as an extension of Chatterjee's theorem \cite{Chatterjee2015} to finite-dimensional spin glass models under random. Random field provides a natural perturbation to break spin-flip symmetry, suppresses thermal fluctuations and favors locally stable spin configurations. Understanding whether the overlap remains self-averaging in such a setting is directly related to the possibility of RSB in finite-dimensional spin glasses. The proof relies on a representation of the EA order parameter through derivatives of the free energy density and on Tasaki's correlation inequality \cite{HT}.  We exploit several properties of the free energy density, such as concavity, 
differentiability almost everywhere in the coupling constant space and independence of boundary conditions. As in the proofs in \cite{Chatterjee2015,I1}, the proof needs the concavity of the free energy density
and its differentiability almost everywhere. Especially for the random field EA model, it is essential that the free energy density is independent of boundary conditions.
  The key idea of the proof is based on the representation of the EA order parameter in terms of the derivative of the free energy density with respect to the strength of random field. 
   This representation of the EA order parameter differs from the one used in \cite{I1}, since the random field preserves the replica symmetry.   Also in the present case, Tasaki's correlation inequality \cite{HT} 
   to prove a lemma in \cite{I1} is utilized again to prove that the upper bound on the expectation of the squared replica overlap is written in terms of the EA order parameter. 
   Then, the replica overlap is self-averaging
 in finite-dimensional spin glass models under Gaussian random field almost everywhere in the coupling constant space in any dimension, unlike RSB phase in the SK model under random field.
 The random field EA model behaves differently from mean-field spin glass models,
 since the self averaging violation of the replica overlap is well-known 
 as the Almeida-Thouless (AT) instability  \cite{AT,Che} in the SK model under magnetic field.
 
 We further establish self-averaging of the replica bond overlap in the EA model with Gaussian couplings and without random fields.
In this case, the inverse temperature plays the same role as the random field strength.
This result supports several published claims for  the triviality of the replica bond overlap 
\cite{KM,NS,NewmanStein1997,NewmanStein2002}.
Since the variance of replica bond overlap 
does not vanish in RSB phase in the SK model in the infinite-volume limit \cite{AfCh,P0,Pn,T1,T0,T}, 
we show different nature of
the finite-dimensional spin glass models from that of the mean-field model also without random field.  
Then, we suggest limitations of the empirical prediction previously stated.

%%%%%%%%%%%%%%%%%%
\section{Definitions}
 \label{s:Def}
%%%%%%%%%%%%%%%%%%
\subsection{Single system}
\label{s:single}
For $d=1,2,\ldots$, let us regard $\bbZ^d$ as the infinite $d$-dimensional hyper cubic lattice, and denote its elements, i.e. sites, as $x,y\ldots\in\bbZ^d$.
The distance between two sites $x,y\in\bbZ^d$ is defined as
\eq
|x-y|\coloneqq\snorm{x-y}_1=\sum_{i=1}^d|x_i-y_i|,
\en
where we wrote $x=(x_1,\ldots,x_d)$.
For a positive integer  $L$ , consider a $d$-dimensional 
hyper cubic lattice with a linear size $L$
\eq
\LaL:= [1,L]^d \cap \bbZ^d,
\lb{LaL}
\en
and its boundary
\eq
\dLaL\coloneqq\bigl\{u\in\bbZ^d\bs\LaL\,\bigl|\,|u-x|=1\ \text{for some}\ x\in\LaL\bigr\}.
\en
The set of bonds, i.e., unordered sets of neighboring sites in $\LaL$ is denoted as
\eq
\calB_L\coloneqq\bigl\{\{x,y\}\,|\,x,y\in\LaL,\ |x-y|=1\bigr\},
\en
and the set of boundary bonds, i.e., oriented pairs of neighboring sites in $\LaL$ and $\dLaL$ as
\eq
\dcalB_L\coloneqq\bigl\{(x,u)\,\bigl|\, x\in\LaL,\ u\in\dLaL,\ |x-u|=1\,\bigr\}.
\en
See Figure~\ref{f:L7}.
\vspace{-3.3cm}
\begin{figure}[H]
\begin{center}
{\includegraphics[width=7truecm]{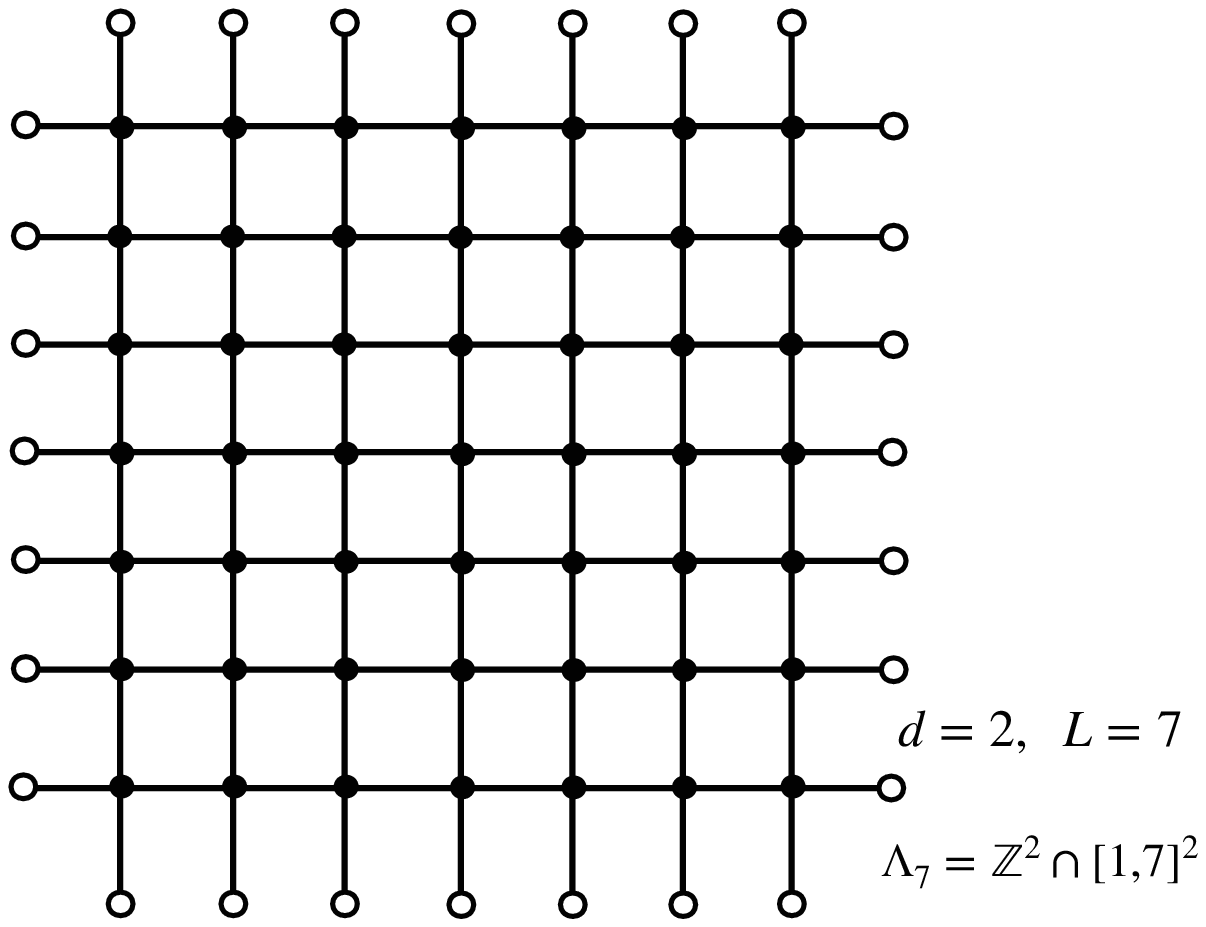}}
\vspace{-2.5cm}
\caption%[dummy]
{
Black dots and white dots represent the sites in $\La_7$ and its boundary $\partial\La_7$, respectively.
Solid lines represent bonds in $\calB_7$ and $\partial\calB_7$.
}
\label{f:L7}
\end{center}
\end{figure}
\vspace{-.8cm}
Assume that an Ising spin described by the spin variable $\sigma_x\in\{1,-1\}$ sits at each site $x\in\LaL$.
A spin configuration, i.e., the collection of all spin variables on $\LaL$, is denoted as $\bss=(\sigma_x)_{x\in\LaL}$, and
$\calC_L$ denotes the set of all spin configurations on $\LaL$.
The boundary spin at each boundary site $u\in\dLaL$ is
described by a continuous variable $b_u\in[-1,1]$.
The collection of all $b_u$ for $u\in\dLaL$ is denoted as $\bsb=(b_u)_{u\in\dLaL}$, and the set of all $\bsb$ is denoted as $\dcalC_L$.

For any $x,y\in\bbZ^d$ such that $|x-y|=1$, let $J_{x,y}=J_{y,x}\in\bbR$ be a random exchange interaction between sites $x$ and $y$, and for any $x\in\bbZ^d$, let $h g_x\in\bbR$ be 
%the standard Gaussian random variable at site $x$ which represents  
a random field at site $x$ with a
field strength $h>0$.
%, and all $J_{x,y}$ are independent and identically distributed
Let  $\bm J:=(J_{x,y})_{\{x,y\}\in {\cal B}_L}$ and $\bm g:=(g_x)_{x\in \Lambda_L }$
be the collection of all $J_{x,y}$ such that $\{x,y\}\in {\cal B}_L$ and all $g_x$ such that   $x\in\LaL$, and $\bbE\,F({\bm J,  \bm g})$ denotes the expectation value  of a function $F({\bm J,  \bm g})$
over random variables  $\bm J$ and $\bm g$. All $\bm g$ 
are  independent and identically distributed (i.i.d.) 
standard Gaussian random variables. All $\bm J$ are i.i.d. random variables, and
the probability distribution of $J_{x,y}$ is arbitrary, except for the assumption
\eq
\bbE\,|J_{x,y}|<\infty.
\lb{EJ}
\en
Then the Hamiltonian of the Edwards-Anderson (EA) model \cite{EA} under  random magnetic fields
 is defined by
\begin{equation}
H_L(\bss;\bm J,h\bm g, \bsb)\coloneqq - \sum_{\{x,y\}\in\calB_L} J_{x,y}\,\sigma_x \sigma_y
-\sum_{(x,u)\in \dcalB_L}(J_{x,u}\,\sigma_x+h~g_u) b_u
- \sum_{x \in \LaL}h~ g_x\,\sigma_x,
\lb{RSHamil}
\end{equation}
where $\bsb\in\dcalC_L$ defines the boundary condition of the spin system.
Standard choices are the open boundary condition with $b_u=0$ for all $u\in\dLaL$ and the plus boundary condition with  $b_u=1$ for all $u\in\dLaL$.

Let us define thermal expectation values and thermodynamic functions.
The thermal average of an arbitrary function $F(\bss)$ of spin configuration $\bss\in\calC_L$
 at  inverse temperature $\beta >0$ is defined by
\eq
\sbkt{F}_{L,\beta;{\bm J, h \bm g},\bsb}
= \frac{1}{Z_L(\beta;{\bm J, h \bm g},\bsb)} \sum_{\bss \in\calC_L} F(\bss)\,e^{ - \beta H_L(\bss;{\bm J, h \bm g},\bsb)},
\lb{<F>}
\en
where the partition function is defined by
\begin{equation}
Z_L(\beta;{\bm J, h \bm g},\bsb) \coloneqq \sum_{\bss \in\calC_L} e^{ - \beta H_L(\bss;{\bm J, h \bm g},\bsb)}.
\end{equation}
For necessity, consider the case
 that the boundary configuration $\bsb$ can depend on the inverse temperature $\beta$ and the interactions 
 ${\bm J, h \bm g}$. In such case, the dependence is represented explicitly as $\bsb(\beta,{\bm J, h \bm g})$.
Define the free energy density by
\begin{equation}
\phi_L(\beta;{\bm J, h \bm g}, \bsb(\beta,{\bm J, h \bm g}))\coloneqq- \frac{1}{\beta L^d}\,
 \log Z_L(\beta;{\bm J, h \bm g},\bsb(\beta,{\bm J, h \bm g})), 
\lb{phiL}
\end{equation}
whose expectation is defined by
\begin{equation}
f_L(\beta,h;\bsb(\cdot))\coloneqq \bbE\,\phi_L(\beta;{\bm J, h \bm g},\bsb(\beta,{\bm J, h \bm g}))
. 
\lb{fL}
\end{equation}
The function $f_L(\beta,h;\bsb(\cdot))$ is concave (or convex-upward) in $h>0$. 
Its infinite-volume limit
\eq
f(\beta,h)\coloneqq\lim_{L\up\infty}f_L(\beta,h;\bsb(\cdot)),
\lb{flim}
\en
exists and is independent of the boundary condition.
See Lemma~\ref{L:f} below.
\\

\begin{lemma}[Infinite-volume limits of the free energy densities]\label{L:f}
The limits \rlb{flim}
exists and is independent of the boundary condition $\bsb(\cdot)$.
\end{lemma}

\vspace{5mm}
The above lemma can be proven by employing the standard technique.
See, e.g., \cite{SimonBook,BovierBook,FV}. It is crucial that the function $f(\beta,h)$ is concave 
(or convex-upward) in $h>0$.
\\

%%%%%%%%%%%%%%%%%%
\subsection{Replicated system}
\label{s:rep}
In order to study order in spin glass models, we introduce $n$ replicated spin systems.
For our purpose, it suffices to consider the cases with $n=2$% and 3
, which are described by the Hamiltonians
\eq
H^{(2)}_L(\bss^{1},\bss^{2};{\bm J, h \bm g},\bsb)\coloneqq\sum_{\alpha=1}^2H_L(\bss^{\alpha};{\bm J, h \bm g},\bsb)
,
\lb{H2}
\en 
where $\bss^{\nu}=(\sigma_x^{\nu})_{x\in\LaL}\in\calC_L$ denotes a spin configuration in the $\nu$-th replica for $\nu=1,2$.
Define the thermal average as
\eqg
\sbkt{F}^{(2)}_{L,\beta;{\bm J, h \bm g},\bsb}
= \frac{1}{Z^{(2)}_L(\beta;{\bm J, h \bm g},\bsb)} \sum_{\bss^{1},\bss^{2}\in\calC_L} F(\bss^{1},\bss^{2})\,e^{ - \beta H^{(2)}_L(\bss^{1},\bss^{2};{\bm J, h \bm g},\bsb)},
\lb{F2}
\eng
with the partition functions
\eqg
Z^{(2)}_L(\beta;{\bm J, h \bm g},\bsb)=\sum_{\bss^{1},\bss^{2}\in\calC_L} 
e^{ - \beta H^{(2)}_L(\bss^{1},\bss^{2};{\bm J, h \bm g},\bsb)}=Z_L(\beta;{\bm J, h \bm g},\bsb) ^2,
\lb{ZL2}
\eng
We again define the averaged free energy for the replicated systems as
\eqg
f^{(2)}_L(\beta,h;\bsb(\cdot))\coloneqq- \frac{1}{\beta L^d}\,\bbE\, \log Z^{(2)}_L(\beta;{\bm J, h \bm g},\bsb(\beta,{\bm J, h \bm g})), 
\lb{fL2}
\eng
Note that the permutation symmetry gives 
\eqg
f^{(2)}_L(\beta,h;\bsb(\cdot))=2f_L(\beta,h;\bsb(\cdot)),
\lb{f2}
\eng
Define the infinite-volume limits of the free energy density also for the two replicated system
\eqg
f^{(2)}(\beta,h)=\lim_{L\up\infty}f^{(2)}_L(\beta,h;\bsb(\cdot))=2f(\beta,h).\lb{f2lim}
\eng

Define an replica overlap between two replicated spin configurations for $\alpha\ne\beta$ by
\eq
R_L^{\alpha,\beta}(\bss^{\alpha},\bss^{\beta})\coloneqq\frac{1}{L^d}\sum_{x\in\LaL}\sigma^\alpha_x\sigma^\beta_x.
\lb{Rab}
\en
The following formula 
for an arbitrary integrable 
function $F(g)$ of a standard Gaussian random variable $ g$ 
is well-known 
\begin{equation}
\mathbb E g F(g) =\int_{-\infty}^\infty \frac{dg}{\sqrt{2\pi}} e^{-\frac{g^2}{2}}gF(g)
=\int_{-\infty}^\infty \frac{dg}{\sqrt{2\pi}}e^{-\frac{g^2}{2}} \frac{d}{dg}F(g)
= \mathbb E F'(g),
\end{equation}
where a relation $\frac{\partial}{\partial g} e^{-\frac{g^2}{2}}=-ge^{-\frac{g^2}{2}}$ and 
integration by parts have been used.
Using this formula, the derivative of the free energy density is
represented in terms of the replica overlap, if the boundary configuration $\bsb(\beta, \bm J)$  is independent of $\bm g$
\begin{eqnarray}
&&\frac{\partial }{\partial h}f_L(\beta,h;\bsb(\cdot)) =-
\frac{1}{|\Lambda_L|}\sum_{x\in \Lambda_L}\mathbb E  g_x 
\langle \sigma_x \rangle_{L,\beta;{\bm J, h \bm g},\bsb(\beta,{\bm J})}-
\frac{1}{|\Lambda_L|}\sum_{u\in \partial \Lambda_L} \mathbb E g_u b_u(\beta,\bm J)
\nonumber \\&&=-\mathbb E 
\frac{1}{|\Lambda_L|}\sum_{x\in \Lambda_L} \frac{\partial}{\partial g_x} 
\langle \sigma_x \rangle_{L,\beta;{\bm J, h \bm g},\bsb(\beta,{\bm J})}
=\beta h \Big(\frac{1}{|\Lambda_L|}\sum_{x\in \Lambda_L}\mathbb E \langle \sigma_x \rangle_{L,\beta;{\bm J, h \bm g},\bsb(\beta,{\bm J})}^2-1\Big) 
\nonumber \\&&= \beta h(\mathbb E \langle R_L^{1,2}\rangle_{L,\beta;{\bm J, h \bm g},\bsb(\beta,{\bm J})}^{(2)}-1).
\lb{fR}
\end{eqnarray}

\begin{lemma}[Chatterjee]\label{Chatterjee} Let $A$ be a subset of the coupling constant space $[0,\infty)^2$, such that the partial derivative $\frac{\partial}{\partial h}f(\beta,h)$ exists at $(\beta,h) \in A$.
The convexity of $f(\beta,h)$ guarantees that its compliment $A^c$ is countable.
The expectation of the replica overlap in the infinite-volume limit is given by the partial
derivative of the free energy density at any $(\beta,h)\in A$
\begin{equation}
\lim_{L\uparrow\infty} \mathbb E \langle R_L^{1,2}\rangle_{L,\beta;{\bm J},h \bm g, \bsb(\beta,{\bm J})}^{(2)}=
\frac{1}{\beta h} \frac{\partial}{\partial h} f(\beta,h)+1,
\lb{Chat}
\end{equation}
if the boundary configuration $\bsb(\beta,{\bm J})$ is independent of the random field $h \bm g$.
\end{lemma}

\vspace{5mm}
Although the above lemma is proven on the basis of the self-averaging property and the concavity
of the free energy density in \cite{Chatterjee2015}, the proof without the self-averaging property
is possible.

%%%%%%%%%%%%%%%%%%%%%%
\section{Order parameters}
\label{s:main}

\subsection{Broadening of the replica overlap}
Note that the EA model \rlb{RSHamil} under  random field
lacks the global $\bbZ_2$ symmetry. 
%In models without $\bbZ_2$ symmetry, the broadening of the distribution of the replica overlap $R_L^{1,2}$
% is regarded as RSB phenomenon.
The broadening of the replica overlap $R_L^{1,2}$ is observed by the standard deviation.
Define an order parameter for broadening as
\eq
\qbr\coloneqq\limsupL\sqrt{
\bbE\,\bbkt{(R_L^{1,2})^2}^{(2)}_{L,\beta;{\bm J, h \bm g},\bm b}
-\Bigl(\bbE\,\bbkt{R_L^{1,2}}^{(2)}_{L,\beta;{\bm J, h \bm g},\bm b}\Bigr)^2
}.
\lb{qbr2}
\en
Chatterjee has proven that $\qbr=0$ 
%defined by \rlb{qbr2}  vanishes 
in the random field Ising model in any dimensions, in any field strength and any temperature
 with open boundary condition $\bm b= \bm 0$
\cite{Chatterjee2015}. Then, Chatterjee has concluded that there is no  RSB phase in the random field Ising model.
This rigorous result has %invalidated
dismissed  several published claims %of self-averaging violation of the replica overlap and 
that RSB phase exists in the random field Ising model \cite{BD,MY}.

\subsection{The Edwards-Anderson order parameter}

A more common order parameter for spin glass is the Edwards-Anderson (EA) order parameter \cite{EA}.
Although the EA order parameter is utilized to detect 
a spontaneous breakdown of the $\bbZ_2$ symmetry, it can be useful to characterize the spin glass order.  
For any $L$, let
\eq
\qEA(L,\beta,h)\coloneqq\bbE\max_{\bsb\in\dcalC_L}\sbkt{R_L^{1,2}}^{(2)}_{L,\beta;{\bm J, h \bm g},\bsb}
=\frac{1}{L^d}\bbE\max_{\bsb\in\dcalC_L}\sum_{x\in\LaL}\bigl(\sbkt{\sigma_x}_{L,\beta;{\bm J, h \bm g},\bsb}\bigr)^2,
\lb{qL}
\en
where we choose a boundary configuration $\bsb$ that maximizes the replica overlap $\sbkt{R_L^{1,2}}^{(2)}_{L,\beta,0;{\bm J, h \bm g},\bsb}$ for each combination of $\beta$, $L$, and ${\bm J, h \bm g}$.
This means $\bsb$ generally depends on $L$, $\beta$, and ${\bm J, h \bm g}$.
Then, the EA order parameter is defined as the infinite-volume limit 
\eq
\qEA(\beta, h)\coloneqq\limL\qEA(L,\beta,h).
\lb{qEA1}
\en
This definition is identical to
the EA order parameter defined by van Enter and Griffits \cite{vEG}.
The existence of the limit is proven at the end of section~\ref{s:proofmain}.
As proven by Chatterjee \cite{Chatterjee2015}, 
the following lemma 
gives another proof of existence of $\qEA(\beta,h)$ with
a representation of the EA order parameter in terms of the derivative of the free energy density.\\

%%%%
%%%%%%%%%%%%%%%
\subsection{Self-averaging of the replica overlap% in the random field EA model
}\label{s:nonZ2}

%\vspace{3mm}
Here, we state our main theorem.\\

\begin{theorem}[Self-averaging of the replica overlap in the random field EA model
]
\label{t:Grgen}
In the random field EA model defined by the Hamiltonian \rlb{RSHamil}, the order parameter \rlb{qbr2}  vanishes at any $(\beta,h)\in A$ for any boundary condition
\eq
\qbr=0.
\lb{Grqjqbr}
\en
\end{theorem}

The theorem states that the replica overlap is self-averaging 
 in the random field EA model everywhere in the coupling constant space
in any dimension. The random field prohibits the broadening of the replica overlap. Then, there is no
phase transition observed by $\qbr$ in the random field EA model as in the random filed Ising model \cite{Chatterjee2015}. The self-averaging of the replica overlap  is
quite natural phenomenon in spin systems under random field, analogous to the magnetization concentration 
in spin systems under uniform field. It is well-known that the magnetization is concentrated always at 
its thermal averaged value in the ferromagnetic Ising model under uniform field.   
Since the SK model has RSB phase even under random field, the theorem provides a counterexample
to the empirical prediction that a finite-dimensional spin system above the lower critical dimension
has topologically equivalent phase diagram of the  corresponding mean-field model.  

\section{Proof}\label{s:proof}
\subsection{Proof of Theorem~ \protect\ref{t:Grgen}}

The following lemmas are essential for our proof.\\

\begin{lemma}[Free energy representation of the EA order parameter]\label{DFED}  
Consider the random field Edwards-Anderson model defined by the Hamiltonian \rlb{RSHamil}. 
The EA order parameter defined by \rlb{qL} and \rlb{qEA1} is
represented in terms of the derivative of the free energy density at all $(\beta,h)\in A$
\eq
\qEA(\beta,h)=\frac{1}{\beta h}\frac{\partial}{\partial h}f(\beta,h)+1.
\lb{qEAf}
\en
\end{lemma}

\begin{lemma}
\label{L:main}
For any boundary condition $\bsb(\cdot)$, the following bound is 
valid for any $(\beta,h)\in A$
\eq
\limsupL\bbE\,\bbkt{(R_L^{1,2})^2}^{(2)}_{L,\beta;\bm J,h \bm g,\bm b(\beta, \bm J,h \bm g)}
\leq  \{\qEA(\beta,h)\}^2.
\lb{main}
\en
\end{lemma}
%The proofs of the lemmas are given in the following.

\medskip
\noindent
{\em Proof of Theorem~\ref{t:Grgen} given Lemma \ref{DFED} and Lemma~\ref{L:main}}\/:
Let $(L_i)_{i=1,2,\ldots}$ be a subsequence that attains the $\limsup$ in \rlb{main}.
Lemma \ref{DFED} implies that the $\limsup$ also in \rlb{qbr2} is obtained by the subsequence 
$(L_i)_{i=1,2,\ldots}$.
Rewrite $\qbr$ defined by \rlb{qbr2} using \rlb{Chat} and \rlb{main}
\eq
(\qbr)^2=\lim_{i\up\infty}\Big[
\bbE\,\bbkt{(R_{L_i}^{1,2})^2}^{(2)}_{L_i,\beta;{\bm J, h \bm g},\bm b}
-\big(\bbE\,\bbkt{R_{L_i}^{1,2}}^{(2)}_{L_i,\beta;{\bm J, h \bm g},\bm b}\big)^2\Big]\leq \{\qEA(\beta,h)\}^2-\Big[\frac{1}{\beta h}\frac{\partial}{\partial h}f(\beta,h)+1\Big]^2.
\lb{qbr3}
\en
The identity \rlb{qEAf}
implies that the right-hand side of \rlb{qbr3} vanishes.
Then, we obtain  \rlb{Grqjqbr}.~\qed\\

Lemma \ref{DFED} implies that 
the EA order parameter converges to the derivative of free energy density 
in the infinite-volume limit in the random field EA model, regardless of the boundary conditions. 
The proof of Lemma~\protect\ref{L:main}  essentially relies on the short-range nature of the model. This means our proof of Theorem ~\ref{t:Grgen} cannot apply to long-range models. 

\subsection{Proof of Lemma \ref{DFED}}
 
The proof relies on the property of the free energy density that the  
derivative $\frac{\partial}{\partial h}f_L(\beta,h;\bsb(\cdot))$ converges to $\frac{\partial}{\partial h}f(\beta,h)$ as $L\up\infty$ for any boundary condition $\bsb(\cdot)$ at any $(\beta,h) \in A$. Let $(\beta,h_0) \in A$ be arbitrarily fixed coupling constants. 
Represent the random field $(g_x)_{x \in \Lambda_L}$  
in terms of i.i.d. standard Gaussian random variables $g^0_x, g^1_x$ with a parameter $t\in \mathbb R$
\begin{equation}
g_x= \frac{h_0 g_x^0 + \sqrt{t} g_x^1}{h(t)},
\end{equation}
where $h(t):=\sqrt{h_0^2+t}$.
Each $g_x$ is independent of another standard Gaussian random variable 
\begin{equation}
g_x^-:= \frac{-\sqrt{t} g_x^0 +h_0 g_x^1}{h(t)}.
\end{equation} 
Inverse transformation is given by
\begin{equation}
g_x^0= \frac{h_0 g_x - \sqrt{t} g_x^-}{h(t)}, \ \ \ g_x^1:= \frac{\sqrt{t} g_x +h_0 g_x^-}{h(t)}.
\end{equation}
To consider the EA order parameter defined by \rlb{qL}, 
let $ \bm b_{\max}(\beta, \bm J,h_0 \bm g^0 )$ be a boundary configuration  
to maximize $\langle R_L^{1,2} \rangle_{L,\beta,\bm J, h_0 \bm g^0, \bm b(\beta,\bm J, h_0 \bm g^0)}$.
Define an another free energy density by
\begin{eqnarray}
e_L(t)&:=&
 \mathbb E \phi_L(\beta, \bm J,h_0 \bm g^0+\sqrt{t} \bm g^1, \bm b_{\max}
 (\beta, \bm J,h_0 \bm g^0 ))\nonumber \\
 &=& \mathbb E \phi_L(\beta, \bm J,h(t) \bm g, \bm b_{\max}
 (\beta, \bm J,h_0 (h_0 \bm g-\sqrt{t} \bm g^-)/h(t) )),
\lb{2random-inf-v}
\end{eqnarray} 
which converges to $f(\beta,h(t))$ in the infinite-volume limit.
Since the boundary configuration $\bm b_{\max}
(\beta, \bm J,h_0 \bm g^0 )$ is independent of random variables 
$\bm g^1$
 in the above free energy density, integration by parts to obtain the formula \rlb{fR} for the expectation of the replica overlap can be used also. The derivative of the above free energy density with respect to $t$
gives 
\begin{eqnarray}
%\frac{\partial}{\partial t}f_L(\beta, h(t), \bm b_{\max}(\cdot))&=&
 e_L'(t) %\nonumber \\&&
&=&-\frac{1}{2\sqrt{t}L^d} \sum_{x\in \Lambda_L} \mathbb E  g_x^1 \langle \sigma_x \rangle_{L,\beta,\bm J,h_0\bm g^0+\sqrt{t}\bm g^1,\bm b_{\max}
(\beta,\bm J,h_0\bm g^0)}\nonumber \\
&=&\frac{\beta}{2} \big( \mathbb E \langle R_L^{1,2}
\rangle_{L,\beta,\bm J,h_0\bm g^0+\sqrt{t}\bm g^1,\bm b_{\max}
(\beta, \bm J,h_0\bm g^0)}^{(2)}-1\big)
\nonumber \\&=&\frac{\beta}{2} \big( \mathbb E \langle R_L^{1,2}
\rangle_{L,\beta,\bm J,h(t)\bm g,\bm b_{\max}
(\beta, \bm J,h_0(h_0 \bm g - \sqrt{t} \bm g^-)/h(t))}^{(2)}-1\big)
%\leq 0
. \lb{2random}
\end{eqnarray}
Then, $\qEA(L,\beta,h_0)$ is identical to $e_L'(0)$
\begin{equation}
\qEA(L,\beta,h_0)=e_L'(0).
\end{equation}
Since $e'_L(t) \in [-\beta/2,0]$ is bounded uniformly, 
there is a subsequence $L_i$ such that 
\begin{equation}
e'(t):=\liminf_{L\uparrow\infty}e_L'(t)= \lim_{i\uparrow\infty}e_{L_i}'(t).
\end{equation}
For any $t_1,t_2\in \mathbb R$, we have
\begin{equation}
\int_{t_1}^{t_2} dt e'(t) =
\lim_{i\uparrow\infty} \int_{t_1}^{t_2} dt e_{L_i}'(t) =\lim_{i\uparrow\infty} [e_{L_i}(t_2)-e_{L_i}(t_1)]
= f(\beta,h(t_2))-f(\beta,h(t_1))=\int_{t_1}^{t_2} dt \frac{\partial}{\partial t}f(\beta,h(t)).
\end{equation}
For almost all $t\in \mathbb R$, $e'(t)$ can be represented in the derivative of the
free energy density
\begin{equation}
e'(t) = \frac{\partial}{\partial t}f(\beta,h(t)).
\end{equation}
The same argument for $\limsup$ as the $\liminf$  implies  
\begin{equation}
\limsup_{L\uparrow \infty}
e_L'(t) =\frac{\partial}{\partial t}f(\beta, h(t)),
\end{equation}
for almost all $t \in \mathbb R$.
Then,
the derivative of the free energy density exists
in the infinite-volume limit
%\begin{equation}
%\lim_{L\uparrow \infty}
%e_L'(t) =\frac{\partial}{\partial t}f(\beta, h(t)).
%\end{equation}
\begin{eqnarray}
h'(t)\frac{\partial}{\partial h}f(\beta, h(t))&=&\lim_{L\uparrow\infty} e_L'(t)
%\frac{\partial}{\partial t} \mathbb E \phi_L(\beta, \bm J,h_0 \bm g^0+\sqrt{t} \bm g^1, \bm b_{\max}(\beta, \bm J,h_0 \bm g^0 ))
\nonumber \\
&=&\lim_{L\uparrow\infty} \frac{\beta}{2}(\mathbb E 
\langle R_L^{1,2}\rangle_{L,\beta,\bm J,h_0\bm g^0+\sqrt{t}\bm g^1,\bm b_{\max}
(\beta, \bm J,h_0 \bm g^0)}^{(2)}-1)
\nonumber \\&=&\frac{\beta}{2} \lim_{L\uparrow\infty} \big( \mathbb E \langle R_L^{1,2}
\rangle_{L,\beta,\bm J,h(t)\bm g,\bm b_{\max}
(\beta, \bm J,h_0(h_0 \bm g - \sqrt{t} \bm g^-)/h(t))}^{(2)}-1\big)
\nonumber \\&=&\frac{\beta}{2} \lim_{L\uparrow\infty} \big( \mathbb E \langle R_L^{1,2}
\rangle_{L,\beta,\bm J,h(t)\bm g,\bm b}^{(2)}-1\big).
\end{eqnarray}
To show the continuity of $e'(t)$ at $t=0$, define a function $d_L(t)$ for $t\in \mathbb R$ by
\begin{equation}
d_L(t) := \frac{\partial}{\partial h_0}
\mathbb E \phi_L(\beta, \bm J,h_0\bm g, \bm b_{\max}
(\beta, \bm J,h_0  (h_0 
\bm g - \sqrt{t} \bm g^-)/h(t)
)). 
\end{equation}
Note that for any finite $L$,
\begin{equation}
\lim_{t\to0}d_L(t) = h'(0)^{-1}e_L'(0).
\end{equation}
For any $\epsilon >0$, there exists $\delta>0$, such that 
\begin{equation}
d_L(t) -\epsilon <  h'(0)^{-1}e_L'(0) < d_L(t) +\epsilon,
\end{equation}
for any $ t \in (-\delta,\delta)$.
The concavity  in $h$  
leads to the following inequality for $t \in [0,\delta)$
\begin{eqnarray}
\frac{\partial}{\partial h}f_L(\beta,h(t),\bm b_{\max}( \cdot)
)
&=&\Big[\frac{\partial}{\partial s}
\mathbb E \phi_L(\beta, \bm J,s \bm g, \bm b_{\max}
(\beta, \bm J,h_0  (h_0 
\bm g - \sqrt{t} \bm g^-)/h(t)
)
) \Big]_{s=h(t)}\nonumber \\
&=&
\Big[\frac{\partial}{\partial s}
\mathbb E \phi_L(\beta, \bm J,s \bm g, \bm b_{\max}
(\beta, \bm J,h_0  (h_0 
\bm g - \sqrt{t} \bm g^-)/h(t)
)
) \Big]_{s=h(t)}\nonumber \\
&\leq&
\Big[\frac{\partial}{\partial s}
\mathbb E \phi_L(\beta, \bm J,s\bm g, \bm b_{\max}
(\beta, \bm J,h_0  (h_0 
\bm g - \sqrt{t} \bm g^-)/h(t)
)) \Big]_{s=h_0}\nonumber \\
&=&  d_L(t) < h'(0)^{-1}
e_L'(0) +\epsilon.
\end{eqnarray}
In limits $L\uparrow \infty$ for any $t \in [0,\delta)$ 
with $(\beta, h(t)) \in A$, the above inequality becomes
\begin{equation}
\frac{\partial}{\partial h} f(\beta, h(t)) \leq 
\lim_{L\uparrow \infty} d_L(t)< h'(0)^{-1}\lim_{L\uparrow \infty} e_L'(0)+\epsilon.
\end{equation}
The same argument for $t\in (-\delta,0)$ with $(\beta, h(t)) \in A$ gives 
\begin{equation}
\frac{\partial}{\partial h} f(\beta, h(t))\geq \lim_{L\uparrow \infty} d_L(t)
> h'(0)^{-1}\lim_{L\uparrow \infty} e_L'(0)-\epsilon.
\end{equation}
These lead to the following inequality for any $t \in (-\delta,\delta)$ with $(\beta,h(t))\in A$
\begin{equation}
\Big|h'(0)^{-1}\lim_{L\uparrow \infty} e_L'(0)-
\frac{\partial}{\partial h}f(\beta, h(t))\Big|< \epsilon,
\end{equation}
which implies the continuity at $t=0$.
For the EA order parameter $\qEA(\beta,h_0)$
defined by \rlb{qL} and \rlb{qEA1}, 
 the above  at $t=0$ becomes
\begin{equation}
\frac{1}{\beta h_0}\frac{\partial}{\partial h_0}f(\beta, h_0)=
\qEA(\beta,h_0)-1.
\end{equation} 
Since the choice $(\beta, h_0)\in A$ is arbitrary, this
 implies \rlb{qEAf}.  
$\Box$

\subsection{Proof of Lemma~\protect\ref{L:main}}
\label{s:proofmain}
The inequality \rlb{main} is the same one given by Lemma 4.1 in \cite{I1}, 
although the upper bound on $\limsupL\bbE\,\bbkt{(R_L^{1,2})^2}^{(2)}_{L,\beta;J,h g,\bsb(\beta,{\bm J, h \bm g})}$ in the latter is represented in terms of a RSB perturbed free energy, 
which is not used in the preset paper. 
The proof of the inequality \rlb{main} is done by the same method based on
a simple correlation inequality devised in \cite{HT}.

Fix $L$, take $\ell$ that is much smaller than $L$, and let $\Lal:=[1,\ell]^d \cap \mathbb Z ^d$ 
be the $d$-dimensional  hypercubic lattice 
 as in \rlb{LaL}.
Let $\Lal^\kappa ~(\subset\LaL)$ with $\kappa=1,\ldots,K$ be the translated copies of $\Lal$, such that
 $|x-y|\ge2$ for all $x\in\Lal^\kappa$ and $y\in\Lal^{\kappa'}$ for any $\kappa\ne\kappa'$.
For notational convenience, assume that copies of $\Lal$ do not touch the boundary of $\LaL$, i.e., $|x-u|\ge2$ for any $x\in\cup_{\kappa=1}^K\Lal^\kappa$ and $u\in\dLaL$.\\
See Figure~\ref{f:L9}.
\vspace{-.3cm}
\begin{figure}[H]
\begin{center}
{\includegraphics[width=6truecm,angle=-90]{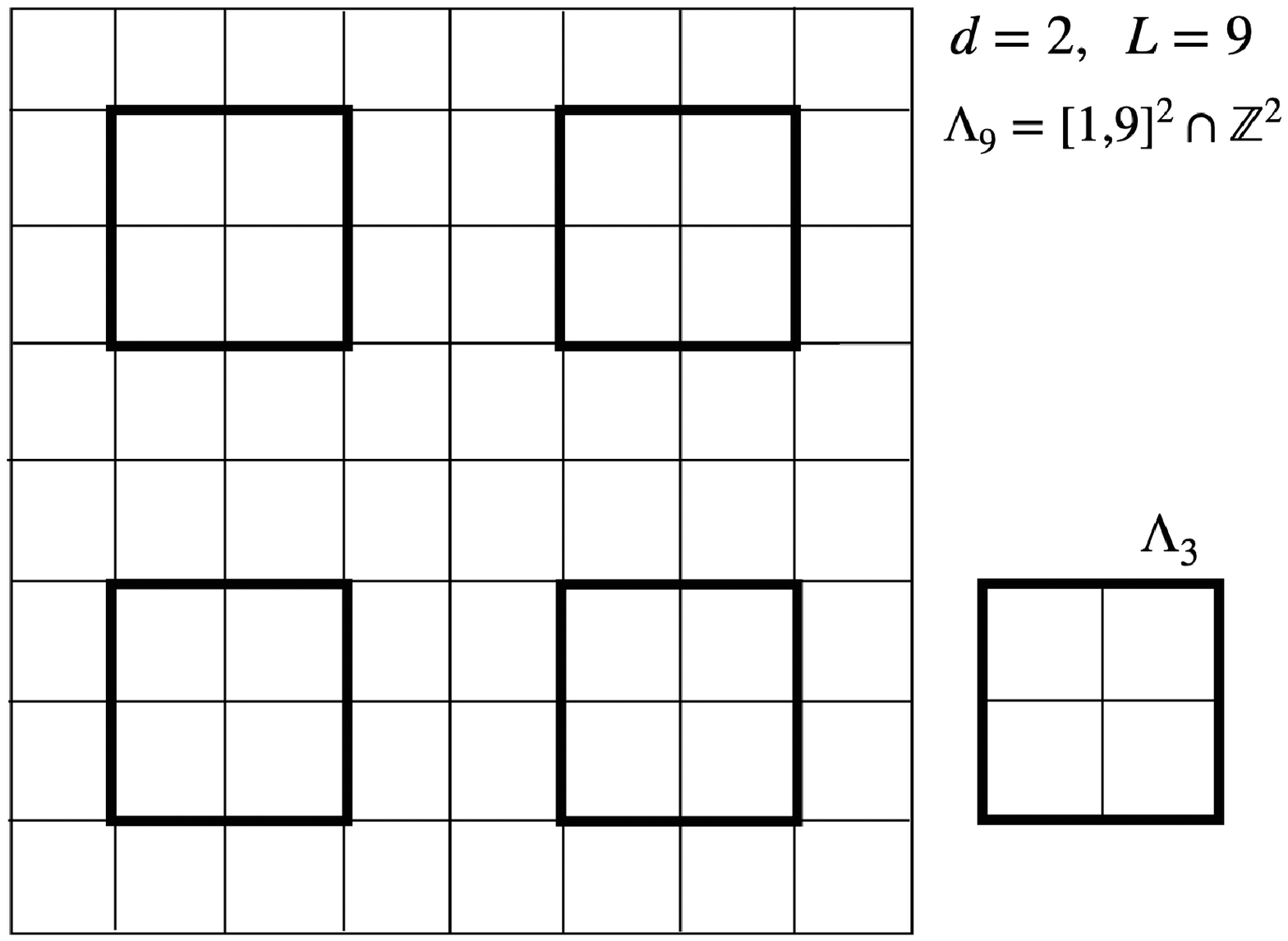}}
%\vspace{-.9cm}
\caption[dummy]{
Four translated copies of $\La_3$ are embedded into $\La_9$.
This corresponds to the optimal choice \rlb{Kopt}.
}
\label{f:L9}
\end{center}
\end{figure}
\vspace{-.5cm}

Fix any $\kappa\ne\kappa'$ and abbreviate $\Lal^\kappa$ and $\Lal^{\kappa'}$ as $\Laa$ and $\Lab$, respectively.
Decompose any spin configuration $\bss=(\sigma_x)_{x\in\LaL}\in\calC_L$ into $\bss=(\bssa,\bssb,\bstau)$, where $\bssa=(\sigma_x)_{x\in\Laa}$, $\bssb=(\sigma_x)_{x\in\Lab}$, and $\bstau=(\tau_x)_{x\in\LaL\backslash(\Laa\cup\Lab)}$. Here, we denote $\tau_x=\sigma_x$ at $x\in\LaL\backslash(\Laa\cup\Lab)$ for later convenience.

Let us tentatively fix a boundary condition $\bsb$ and a random realization of $J_{x,y}$ and $hg_x$, and decompose the Hamiltonian \rlb{RSHamil} as
\eq
H_L(\bss;\bsb)=H_\mathrm{a}(\bssa,\bstau)+H_\mathrm{b}(\bssb,\bstau)+\tilde{H}(\bstau,\bsb),
\lb{HHH}
\en
where, for $\alpha=\mathrm{a}, \mathrm{b}$, we set
\eq
H_\alpha(\bss_\alpha,\bstau)=- \sumtwo{\{x,y\}\in\calB_L}{{\rm s.t.}\,x,y\in\La_\alpha} J_{x,y}\,\sigma_x \sigma_y
-\sumtwo{\{x,y\}\in\calB_L}{{\rm s.t.}\,x\in\La_\alpha,\,y\in\LaL\backslash\La_\alpha}J_{x,y}\,\sigma_x\tau_y
- \sum_{x \in \La_\alpha}hg_x\,\sigma_x.
\en
It is crucial that $\tilde{H}$ does not depend on $\bssa$ or $\bssb$.
We have tentatively dropped the ${\bm J, h \bm g}$ dependence for notational simplicity.

Take any $x\in\Laa$ and $y\in\Lab$.
Consider the thermal average of $\sigma_x\sigma_y$ as defined in \rlb{<F>}, and rewrite it as
\eqa
\sbkt{\sigma_x\sigma_y}_{L,\beta;\bsb}&=\frac{1}{Z_{L,\beta;\bsb}}\sum_{\bss\in\calC_L}\sigma_x\sigma_y\,e^{-\beta H_L(\bss;\bsb)}
\nl
&=\frac{1}{Z_{L,\beta;\bsb}}\sum_{\bstau}e^{-\beta\tilde{H}(\bstau;\bsb)}
\sum_{\bssa}\sigma_x\,e^{-\beta H_\mathrm{a}(\bssa,\bstau)}
\sum_{\bssb}\sigma_y\,e^{-\beta H_\mathrm{b}(\bssb,\bstau)}
\nl
&=\frac{1}{Z_{L,\beta;\bsb}}\sum_{\bstau}e^{-\beta\tilde{H}(\bstau;\bsb)}\,
Z_{\mathrm{a},\bstau}\,\sbkt{\sigma_x}_{\mathrm{a},\bstau}\,
Z_{\mathrm{b},\bstau}\,\sbkt{\sigma_y}_{\mathrm{b},\bstau}
\lb{sxsy1}
\ena
where the thermal average of sub systems for $\alpha=\mathrm{a}, \mathrm{b}$ is defined by
\eq
\sbkt{\cdots}_{\alpha,\bstau}=\frac{1}{Z_{\alpha,\bstau}}\sum_{\bss_\alpha}(\cdots)\,e^{-\beta H_\alpha(\bss_\alpha,\bstau)},\quad
Z_{\alpha,\bstau}=\sum_{\bss_\alpha}e^{-\beta H_\alpha(\bss_\alpha,\bstau)}.
\en
Define a distribution of spin configuration $\bm \tau$ by
\eq
P(\bstau)=\frac{e^{-\beta\tilde{H}(\bstau;\bsb)}\,Z_{\mathrm{a},\bstau}\,Z_{\mathrm{b},\bstau}}{Z_{L,\beta;\bsb}},
\en
which satisfies $P(\bstau)\ge0$ and $\sum_{\bstau}P(\bstau)=1$. Then \rlb{sxsy1} can be written as
\eq
\sbkt{\sigma_x\sigma_y}_{L,\beta;\bsb}=\sum_{\bstau}P(\bstau)\,\sbkt{\sigma_x}_{\mathrm{a},\bstau}\,\sbkt{\sigma_y}_{\mathrm{b},\bstau}.
\en
 Let us evaluate the summation of the squared two point function over all sites in $\Lambda_\mathrm{a}$ and
  $\Lambda_\mathrm{b}$, using the above representation
\eqa
\sumtwo{x\in\Laa}{y\in\Lab}\bigl(\sbkt{\sigma_x\sigma_y}_{L,\beta;\bsb}\bigr)^2&=\sumtwo{x\in\Laa}{y\in\Lab}\sum_{\bstau,\bstau'}P(\bstau)\,P(\bstau')\,\sbkt{\sigma_x}_{\mathrm{a},\bstau}\,\sbkt{\sigma_x}_{\mathrm{a},\bstau'}\,\sbkt{\sigma_y}_{\mathrm{b},\bstau}\,\sbkt{\sigma_y}_{\mathrm{b},\bstau'}
\nl&\le\sumtwo{x\in\Laa}{y\in\Lab}\sum_{\bstau}P(\bstau)\,\bigl(\sbkt{\sigma_x}_{\mathrm{a},\bstau}\bigr)^2\,\bigl(\sbkt{\sigma_y}_{\mathrm{b},\bstau}\bigr)^2
\nl&=\sum_{\bstau}P(\bstau)\Bigl\{\sum_{x\in\Laa}\bigl(\sbkt{\sigma_x}_{\mathrm{a},\bstau}\bigr)^2\Bigr\}
\Bigl\{\sum_{y\in\Lab}\bigl(\sbkt{\sigma_y}_{\mathrm{b},\bstau}\bigr)^2\Bigr\},
\lb{sxsy2}
\ena
where we have used the trivial inequality
\eq
\sbkt{\sigma_x}_{\mathrm{a},\bstau}\,\sbkt{\sigma_x}_{\mathrm{a},\bstau'}\,\sbkt{\sigma_y}_{\mathrm{b},\bstau}\,\sbkt{\sigma_y}_{\mathrm{b},\bstau'}
\le\frac{1}{2}\Bigl\{\bigl(\sbkt{\sigma_x}_{\mathrm{a},\bstau}\bigr)^2\,\bigl(\sbkt{\sigma_y}_{\mathrm{b},\bstau}\bigr)^2+\bigl(\sbkt{\sigma_x}_{\mathrm{a},\bstau'}\bigr)^2\,\bigl(\sbkt{\sigma_y}_{\mathrm{b},\bstau'}\bigr)^2\Bigr\},
\en
to get the second line.

Since a part of the configuration $\bstau$ plays the role of boundary condition for the expectation value $\sbkt{\cdots}_{\alpha,\bstau}$,
we obtain Tasaki's inequality \cite{HT}, which  gives an upper-bounded on the final line in \rlb{sxsy2}  
\eq
\sumtwo{x\in\Laa}{y\in\Lab}\bigl(\sbkt{\sigma_x\sigma_y}_{L,\beta;{\bm J, h \bm g},\bsb}\bigr)^2
\le\max_{\bsb'}\sum_{x\in\Laa}\bigl(\sbkt{\sigma_x}_{\mathrm{a};{\bm J, h \bm g},\bsb'}\bigr)^2\,\max_{\bsb''}\sum_{x\in\Lab}\bigl(\sbkt{\sigma_x}_{\mathrm{b};{\bm J, h \bm g},\bsb''}\bigr)^2.
\lb{sxsy3}
\en
Here $\sbkt{\cdots}_{\mathrm{a};{\bm J, h \bm g},\bsb'}$ denotes the expectation value, exactly as in \rlb{<F>}, on the lattice $\Laa$ with boundary configuration $\bsb'$ and random interactions and fields determined by 
${\bm J, h \bm g}$.

%We now allow the boundary configuration $\bsb$ on $\partial \LaL$ to depend on ${\bm J, h \bm g}$.
If $ \Lambda_\ell^\kappa
\cap\Lambda_\ell^{\kappa'}=\phi$
for any $\kappa,\kappa'=1,\ldots,K$, then
 the sample expectation of \rlb{sxsy3} over $\bm J, \bm g$ enables us to
 get our main inequality
 \eq
\bbE\sum_{x\in\Lal^\kappa,\,y\in\Lal^{\kappa'}}\bigl(\sbkt{\sigma_x\sigma_y}_{L,\beta;{\bm J, h \bm g},\bsb(\beta,{\bm J, h \bm g})}\bigr)^2
\le\ell^{2d}\,\bigl\{\qEA(\ell,\beta, h)\bigr\}^2
\lb{sxsy4}
\en
 where $\qEA(\ell,\beta, h)$ is defined in \rlb{qL}.
The two expectation values on the right-hand side of \rlb{sxsy3} are independent,
if $ \Laa\cap  \Lab=\phi$.

To complete the proof, consider the decomposition of the summation over all lattice sites $x,y\in \Lambda_L$
\eq
\sum_{x,y\in\LaL}(\cdots)=\mathop{\sum_{\kappa,\kappa'=1}^K}_{{\rm s.t.}\,\kappa\neq{\kappa'} }\sum_{x\in\Lal^\kappa,\,y\in\Lal^{\kappa'}}(\cdots)
+\sum_{(x,y)\in\mathcal{R}}(\cdots),
\lb{sumdec}
\en
where $\mathcal{R}$ is defined by the above equation.
The set $\mathcal{R} (\subset \Lambda_L^2)$ is defined by 
a set of pairs $x$ and $y$, such that 
$x,y\in\Lal^\kappa\cup \Lal^{\kappa'}$ for some
$\kappa, \kappa'$ satisfying
$\Lal^\kappa\cap \Lal^{\kappa'} \neq \phi$.
It is crucial to us that the first sum is over $K(K-1)\,\ell^{2d}$ terms and the second sum is over $L^{2d}-K(K-1)\,\ell^{2d}$ terms.
We then find
\eq
\bbE\sum_{x,y\in\LaL}\bigl(\sbkt{\sigma_x\sigma_y}_{L,\beta;{\bm J, h \bm g},\bsb(\beta,{\bm J, h \bm g})}\bigr)^2
\le K(K-1)\,\ell^{2d}\,\bigl\{\qEA(\ell,\beta, h)\bigr\}^2+\bigl\{L^{2d}-K(K-1)\,\ell^{2d}\bigr\},
\lb{sxsy5}
\en
where we have used \rlb{sxsy4} for the first sum in the right-hand side of \rlb{sumdec} and $\bigl(\sbkt{\sigma_x\sigma_y}_{L,\beta;{\bm J, h \bm g},\bsb(\beta,{\bm J, h \bm g})}\bigr)^2\le1$ for the second sum.
Note that the optimal choice of the number of translated copies is
\eq
K=\Bigl\lfloor\frac{L-1}{\ell+1}\Bigr\rfloor^d,
\lb{Kopt}
\en
which implies
\eq
\limL\frac{K}{L^d}=\frac{1}{(\ell+1)^d}.
\en
We then find from \rlb{sxsy5} that
\eqa
\limsupL\bbE\,\bbkt{(R_L^{1,2})^2}^{(2)}_{L,\beta;{\bm J, h \bm g},\bsb(\beta,{\bm J, h \bm g})}&=
\limsupL\frac{1}{L^{2d}}\,\bbE\sum_{x,y\in\LaL}\bigl(\sbkt{\sigma_x\sigma_y}_{L,\beta;{\bm J, h \bm g},\bsb(\beta,{\bm J, h \bm g})}\bigr)^2
\nl
&\le\Bigl(\frac{\ell}{\ell+1}\Bigr)^{2d}\,\{\qEA(\ell,\beta, h)\}^2+\Bigl\{1-\Bigl(\frac{\ell}{\ell+1}\Bigr)^{2d}\Bigr\}, \lb{BoundR2}
\ena
for any $\ell$. Lemma \ref{DFED} guarantees the existence of $\qEA(\beta, h)
=\lim_{\ell \up \infty}\qEA(\ell,\beta, h)$
, then we get 
\eq
\limsupL\bbE\,\bbkt{(R_L^{1,2})^2}^{(2)}_{L,\beta;{\bm J, h \bm g},\bsb(\beta,{\bm J, h \bm g})}\le
\{\qEA(\beta, h)\}^2,
\lb{R2qEA2}
\en
which is the desired \rlb{main}.  \qed\\

\section{Discussions}
\noindent
\subsection{Another proof for the existence of the EA order parameter}
 It is possible to show the existence of the limit without Lemma \ref{DFED}. Take $\liminf_{\ell\up\infty}$ in the right-hand side in \rlb{BoundR2}, and
let $\bsb_\mathrm{max}(\beta,{\bm J, h \bm g})$ be a boundary configuration that provides the maximum in \rlb{qL}.
Then, using the non-negativity of variance twice, we see 
\eqa
\bbE\,\bbkt{(R_L^{1,2})^2}^{(2)}_{L,\beta;{\bm J, h \bm g},\bsb_\mathrm{max}(\beta,{\bm J, h \bm g})}
&\ge
\bbE\,\bigl(\bbkt{R_L^{1,2}}^{(2)}_{L,\beta;{\bm J, h \bm g},\bsb_\mathrm{max}(\beta,{\bm J, h \bm g})}\bigr)^2
\nl&\ge
\bigl(\bbE\,\bbkt{R_L^{1,2}}^{(2)}_{L,\beta;{\bm J, h \bm g},\bsb_\mathrm{max}(\beta,{\bm J, h \bm g})}\bigr)^2
=\{\qEA(L,\beta, h)\}^2.
\lb{R2qEA2B}
\ena
By taking $\limsupL$ and using \rlb{R2qEA2}, we find
\eq
\limsupL\{\qEA(L,\beta, h)\}^2\le\liminfL\{\qEA(L,\beta, h)\}^2.
\en
Since $\qEA(L,\beta, h)\ge0$, we see that $\limL\qEA(L,\beta, h)$ exists. \qed

\subsection{The Edwards-Anderson model without random field}

In the present model for the perturbation parameter fixed at $h=0$,  
Lemma \ref{DFED} cannot be proven, 
and therefore Theorem \ref{t:Grgen} cannot be proven either. 
For a replica bond overlap, however, a theorem corresponding to Theorem \ref{t:Grgen} can be
proven.
 Consider the EA model with the exchange interactions $\bm J:=(J_{x,y})_{\{x,y\}\in {\cal B}_L}$ satisfying 
i.i.d. standard Gaussian distribution. 
 Define a replica bond overlap between different replica spin configurations with replica indices 
 $\alpha\neq\beta$  by
\begin{equation}
Q^{\alpha,\beta}_L:=\frac{1}{L^d d} \sum_{\{x,y \} \in {\cal B}_L}
\sigma_x^\alpha\sigma_y^\alpha\sigma_x^\beta\sigma_y^\beta.
\end{equation}
We have the following representation corresponding to Lemma \ref{DFED}
\begin{equation}
\lim_{L\uparrow \infty}\mathbb E\langle Q^{1,2}_L\rangle_{L,\beta;\bm J, \bm 0, \bm b} = 
\frac{1}{\beta d}\frac{\partial}{\partial \beta} [\beta f(\beta,0)]+1.
\end{equation}
Since the function $\beta f(\beta,0)$ is a convex function of $\beta$, the inverse temperature $\beta$ plays the same role of $h$ as in the proof of Lemma \ref{DFED}.
Therefore, it is proven  
that its variance vanishes in the infinite-volume limit for almost all $\beta>0$
\begin{equation}
\lim_{L\uparrow \infty}\big[\mathbb E\langle (Q^{1,2}_L)^2 \rangle_{L,\beta;\bm J, \bm 0, \bm b}- (\mathbb E \langle Q^{1,2}_L\rangle_{L,\beta;\bm J, \bm 0, \bm b})^2\big]=0,
\label{vbo}
\end{equation}
as in the proof of  Theorem \ref{t:Grgen}. There is no phase transition observed by the above variance 
in the EA model. This result supports several published arguments for the self-averaging of the replica bond overlap \cite{FH, FH2,KM,NS,NewmanStein1997,NewmanStein2002}. 
On the other hand in the SK model with $N$ spins, 
the replica bond overlap  can be written in terms of the replica overlap
\begin{equation} 
Q^{\alpha,\beta}_N =  \frac{2}{N(N-1)} \sum_{1\leq x<y\leq N}
\sigma_x^\alpha\sigma_y^\alpha\sigma_x^\beta\sigma_y^\beta= \frac{N}{N-1}\Big[(R^{\alpha,\beta}_N)^2-\frac{1}{N}\Big].
\end{equation}
In sufficiently low temperature in the SK model, 
the replica overlap distribution is not supported on a single value of $Q^{1,2}_N=(R^{1,2}_N)^2$ in the infinite-volume limit \cite{AfCh,P0,Pn,T1,T0,T}. Then, the variance of $Q^{1,2}_N$  
does not vanish in RSB phase in the SK model in the infinite-volume limit, unlike our result (\ref{vbo}) in the EA model. 
The present paper indicates that natures of
finite-dimensional spin glass models in any finite dimension differ from those of mean-field spin glass models.
 Although we formulate the theorem for the EA model, the argument can be extended to a broader class of finite-dimensional spin glass models.    

%%%%%%%%%%%%%%%%%%%%%%%%%%%%%%%%%%%
%%%%%%%%%%%%%%%%%%%%%%%%%%%%%%%%%%%
%\newpage
\paragraph*{Acknowledgments}

We express our sincere gratitude to H. Tasaki for critical and helpful comments regarding our view on the result. We are grateful to K. Hukushima for shearing insight based on his extensive research experience in spin glasses. 
We thank %Aernout van Enter, K. Hukushima, 
H. Mukaida and N. Shiraishi 
%, Harukuni Ikeda, Hidetoshi Nishimori, Yoshinori Sakamoto, Masafumi Udagawa, and Mizuki Yamaguchi 
for their valuable discussions.
The present research is supported by JSPS Grants-in-Aid for Scientific Research Nos.~25K07176.

\medskip
\noindent
\paragraph*{Conflict of interest statement}
The authors declare no conflicts of interest.

\paragraph*{Data availability statement}
Data sharing not applicable to this article as no datasets were generated or analyzed
during the current study and article describes entirely theoretical research.

\end{document}